\documentclass{llncs}
\usepackage{amsmath}
\usepackage{amsfonts}
\usepackage{amssymb}
\usepackage{graphicx}
\usepackage{fourier}
\usepackage[ruled,vlined]{algorithm2e}
\usepackage{float}
\makeatletter
\renewcommand\subsubsection{\@startsection{subsubsection}{3}{\z@}%
                       {-18\p@ \@plus -4\p@ \@minus -4\p@}%
                       {0.5em \@plus 0.22em \@minus 0.1em}%
                       {\normalfont\normalsize\bfseries\boldmath}}
\makeatother
\title{Chauffeuring a Crashed Robot from a Disk}
\author{Debasish Pattanayak, H. Ramesh, and Partha Sarathi Mandal}
\institute{Indian Institute of Technology Guwahati, India
	\\\email{\{p.debasish,ramesh\_h,psm\}@iitg.ac.in}}
\begin{document}
\maketitle
\begin{abstract}
	Evacuation of robots from a disk has attained a lot of attention recently. We visit the problem from the perspective of fault-tolerance. We consider two robots trying to evacuate from a disk via a single hidden exit on the perimeter of the disk. The robots communicate wirelessly. The robots are susceptible to crash faults after which they stop moving and communicating. We design the algorithms for tolerating one fault. The objective is to minimize the worst-case time required to evacuate both the robots from the disk.
	When the non-faulty robot chauffeurs the crashed robot, it takes $\alpha \geq 1$ amount of time to travel unit distance. With this, we also provide a lower bound for the evacuation time. 
	Further, we evaluate the worst-case of the algorithms for different values of $\alpha$ and the crash time.
\end{abstract}

\section{Introduction}
Searching has always been a classical problem and by extension, the search of a hidden object in a domain has piqued the interest. In particular, a variant of search problem introduced by Czyzowicz et al.~\cite{CzyzowiczGGKMP14} as the \textit{evacuation} problem, which tries to minimize the time required for the last searcher to reach the target. The searchers, in this case, are mobile robots which can move around in the domain. The domain considered can be a convex shape like a disk, a circle or a triangle~\cite{BrandtLLSW17,ChuangpishitGS18,ChuangpishitMNO17}, while the target is hidden on the boundary. Another variant of the problem considers the domain as lines and rays~\cite{BrandtFRW17} with faulty robots~\cite{CzyzowiczKKNO16,KupavskiiW18}. The objective here is to minimize the \textit{competitive ratio} between the time required for the robot to reach the exit and the distance from the exit. 

The recent literature has focused on the aspect of collective-collaborative search. Czyzowicz et al.~\cite{CzyzowiczGGKMP14} introduced two robots trying to search for an exit located on the perimeter of a unit disk. A robot can locate the exit only when it is at the exit. There are two models of communication between the robots, namely, \textit{wireless} and \textit{face-to-face}. In face-to-face, the robots can exchange messages if they are collocated at the same point at the same time. They showed that the evacuation time for two robots in the wireless model was $1+2\pi/3+\sqrt{3}\approx 4.826$, which was optimal. In the same paper, they achieved an upper bound of  $5.740$ and lower bound $5.199$ for two robots in the face-to-face model. In a subsequent paper, Czyzowicz et al.~\cite{CzyzowiczGKNOV15} improved the upper and lower bounds to $5.628$ and $5.255$, respectively. Later Brandt et al.~\cite{BrandtLLSW17}, further improved the upper bound to $5.625$.

Further Czyzowicz et al.~\cite{fun/CzyzowiczGKKKNO18} study priority evacuation of a particular robot from the disk, namely the Queen, while other servant robots search for the exit. This paper sets the upper and lower bounds for evacuation for one, two and three servants. Another paper by Czyzowicz et al.~\cite{sirocco/CzyzowiczGKKKNO18} presents the bounds for $n\geq 4$.

Another aspect of this evacuation problem deals with fault-tolerance. The two types of faults considered are \textit{crash fault} and \textit{byzantine fault}. The type of crash fault considered in~\cite{CzyzowiczGGKKRW17,CzyzowiczKKNO16} does not detect the target when it passes through or does not communicate when it finds the target. The robots with byzantine faults in~\cite{CzyzowiczGGKKRW17,CzyzowiczGKKNOS16} even lie about the position of the target. Czyzowicz et al. \cite{CzyzowiczGGKKRW17} focused on minimizing the evacuation time for the latest non-faulty robot. They achieved a lower bound of $5.188$ and upper bound $6.309$ for three robots out of which at most one is susceptible to crash fault with wireless communication.
We initiate the study on a type of crash fault, where the robot stops moving and sending messages altogether (unlike~\cite{CzyzowiczGGKKRW17}). Instead of abandoning the crashed robot, our objective is to chauffeur it to the exit in the least time. As a natural outcome, we consider the robot which chauffeurs the crashed robots incurs an extra cost. The chauffeur carrying the crashed robot travels at a fraction of its original speed.

\paragraph{\textbf{Our Contributions}:}
In this paper, we consider a variant of crash fault where the affected robot stops moving and communicating after it has crashed. We address the problem of evacuation for two robots out of which at most one can be faulty. 
We propose three evacuation algorithms in the wireless communication model where the non-faulty robot chauffeurs the crashed robot.
\begin{itemize}
	\item We present a lower bound for evacuation time.
	\item We rigorously analyze our algorithms to provide the worst-case evacuation time corresponding to the crash time $w$ at which a robot becomes faulty.
	\item We compare the performance of the algorithms given particular values of the crash time ($w$) and the chauffeuring cost ($\alpha$).
\end{itemize}

\section{Model and Preliminaries}\label{sec:model}
We consider the evacuation of two robots from a unit disk, i.e., a disk with radius one. Let $R_1$ and $R_2$ be the robots. The disk contains an exit located on its perimeter. Both robots have to evacuate the disk. Initially, the robots are situated at the center of the disk and start moving at the same time towards the perimeter of the disk. A robot can find an exit only when it reaches the position of that exit.

We follow a convention that $\widearc{AE}$ denotes the arc along the perimeter on the disk starting at $A$ and ending at $E$ moving in the counter-clockwise direction. Accordingly, $\widearc{EA}$ is the complement of the arc $\widearc{AE}$. 
We abuse the notation $\widearc{AE}$ to denote the length of the arc and $\overline{AE}$ to denote the length of the chord corresponding to the arc $\widearc{AE}$.
Note that, the length of a chord corresponding to an arc of length $\zeta$ is $2\sin(\zeta/2)$.

Both robots travel at a uniform speed of one unit distance per unit time. We consider that the robots are susceptible to \textit{crash faults}. A robot which has been crashed stops moving and communicating thenceforth. At most a single robot is faulty. Suppose, $R_1$ has crashed, $R_2$ can still chauffeur $R_1$ to the exit. A non-faulty robot chauffeuring the crashed robot travels at a speed $1/\alpha$ times the original speed. In other words, it takes $\alpha$ time to travel unit distance, where $\alpha \geq 1$.

The robots can communicate by sending messages wirelessly. The communication is reliable. For analytical purposes, we ignore the message propagation delay. The robots communicate frequently with each other. Without loss of generality, let $w$ be the time after which $R_1$ crashes. So $w-1$ is the distance travelled by the robot after reaching the perimeter at a speed of one unit distance per unit of time before it crashes at time $w$. We also follow a convention that $x$ is the distance of the exit in the counter-clockwise direction starting from the point where $R_1$ has reached the perimeter. So, $\widearc{AE} = x$ as shown in Fig.~\ref{fig:MoveTogether}. We assume that $O$ is the origin and $\overline{OA}$ is the positive $x$-axis.

\paragraph{\textbf{Evacuation Problem (2,1)-crash fault}:} The objective is to minimize the time required by the latest robot to evacuate from the unit disk via an exit located on the perimeter of the disk starting from the center of the disk, where both robots travel at uniform speed of one unit distance per unit time and at most one robot is faulty out of the two. Chauffeuring the faulty robot increases the time for movement by a factor $\alpha \geq 1$ for the non-faulty robot.

\section{Lower Bounds for Wireless Communication}\label{sec:lowerbound}
The lower bound for wireless communication model without faults is $ 1 + 2\pi/3 + \sqrt{3} \approx 4.826$~\cite{CzyzowiczGGKMP14}. Hence, the lower bound is applicable to the crash fault model if the crash time $w \geq 1 + 2\pi/3 + \sqrt{3}$. 
We have the following theorems if the robot crashes before evacuation.

\begin{theorem}\label{thm:wlessthan1}
	The lower bound for evacuation with crash fault for crash time $ w < 1$ is $2\pi + w + \alpha(1-w)$.
\end{theorem}

\begin{proof}
	Consider the simple case where the robot crashes at the center of the disk immediately after it is activated.  As the adversary can always place the exit at a position which is still unexplored, it requires at least $2\pi$ time to search for the exit on the perimeter by the non-faulty robot. Additionally, $\alpha$ is the minimum time required for the faulty robot to be carried to the exit from the center. Since the non-faulty robot can also travel to the perimeter while carrying the faulty robot, it does not need additional time to reach the perimeter. Then the worst-case time for evacuation for the faulty robot is always greater than $2\pi + \alpha$.
	Similarly, if $0<w\leq 1$, we have the lower bound for evacuation at 
	\begin{equation}
		2\pi + w + \alpha(1-w)
	\end{equation}
	where the robots travel together for a distance $w$, and the non-faulty robot chauffeurs the faulty robot rest of the $1 -w$ distance to the perimeter.
\end{proof}

\begin{lemma}[Lemma 5 from \cite{CzyzowiczGGKMP14}]\label{lem:unexplored}
	Consider a perimeter of a disk whose subset of total length $u + \epsilon > 0$ has not been explored for some $\epsilon > 0$ and $\pi \geq u > 0$. Then there exist two unexplored boundary points between which the distance along the perimeter is at least $u$.
\end{lemma}
Intuitively, the proof follows from this argument. If $u$ is the unexplored part of the perimeter and it is a continuous arc, then the endpoints of the arc are separated by at least the same distance along the perimeter.
Otherwise, if the unexplored part is separated into multiple arcs, then the distance between two farthest endpoints of the two arcs are at least separated by a distance which is greater than the unexplored part $u \leq \pi$. A detailed proof is available in~\cite{CzyzowiczGGKMP14}.

\begin{theorem}\label{thm:wgreaterthan1plus2pi3}
    The lower bounds for evacuation with crash fault for crash time $w > 1 + 2\pi/3$ are
    \begin{itemize}
        \item $1 + 2\pi/3 + \sqrt{3}$ for $w \geq 1 + 2\pi/3 + \sqrt{3}$
        \item $1 + 2\pi/3 + (\alpha+1)(1 + 2\pi/3 + \sqrt{3} -w)$ for $w \in [1 + 2\pi/3, 1 + 2\pi/3 + \sqrt{3}/2]$
        \item $1 + 2\pi/3 + \sqrt{3} + (\alpha-1)(1 + 2\pi/3 + \sqrt{3} - w)$ for $w \in [1 + 2\pi/3 + \sqrt{3}/2, 1 + 2\pi/3 + \sqrt{3}]$
    \end{itemize}
\end{theorem}

\begin{proof}
    If $w > 1 + 2\pi/3 + \sqrt{3}$, then we claim that the robots evacuate before any of them fails. By the time $1 + t$, two robots can explore at most $2t$ on the perimeter. The unexplored part is $2\pi - 2t$. According to Lemma~\ref{lem:unexplored}, there exist two unexplored points such that the smallest arc between them is at least $2\pi - 2t$. If a robot is at one of the unexplored points, then the adversary can place the exit at the other unexplored point and the exit is found. The robot, which does not find the exit, can receive a message and travel to the exit along a straight line through the disk. The evacuation time in that case is $ 1 + t + 2\sin(t)$ over all possible values of $\pi/2 < t< \pi$. 
    This results in a worst-case at $ t = 2\pi/3$, which is $ 1 + 2\pi/3 + \sqrt{3} \approx 4.826$.
    Since, the crash time $w$ is more than the worst-case evacuation time, both robots evacuate before the fault occurs and the lower bound for evacuation is $1 + 2\pi/3 + \sqrt{3}$.

    If $w < 1 + 2\pi/3 + \sqrt{3}$, the robot can fail at any moment before the evacuation occurs. For this part of the proof, consider that the exit is found at time $1 + 2\pi/3$ by $R_2$.
    In the worst-case $R_1$ travels a distance $\sqrt{3}$ from its own position towards the exit in a straight line. The distance between the two robots is $\sqrt{3}$, so the robots can meet at the midpoint of that chord at time $1 + 2\pi/3 + \sqrt{3}/2$. 
    If $R_1$ fails before the midpoint, then $R_2$ has to travel extra distance to chauffeur $R_1$. Otherwise, they meet at the midpoint and travel together. In that case, if $R_1$ fails, $R_2$ chauffeurs the remaining distance. The following two cases present the evacuation time if the robot fails before or after the midpoint.
            
    \begin{description}
        \item[Fault before the midpoint ($1 + 2\pi/3< w < 1 + 2\pi/3 + \sqrt{3}/2$):] $R_2$ has to travel extra distance of $1 + 2\pi/3 + \sqrt{3}/2 -w $ to reach $R_1$ from midpoint and chauffeur it for a distance $1 + 2\pi/3 + \sqrt{3} - w$. The evacuation time is $1 + 2\pi/3 + (\alpha+1)(1 + 2\pi/3 + \sqrt{3} -w)$.
        \item[Fault after the midpoint($1 + 2\pi/3 + \sqrt{3}/2 < w < 1 + 2\pi/3 + \sqrt{3}$):] $R_1$ and $R_2$ have already met and are travelling together. The distance $R_2$ needs to chauffeur $R_1$ is $ 1 + 2\pi/3 + \sqrt{3} - w$. The chauffeuring distance adds the cost by a factor of $(\alpha -1)$ since the robots were already travelling towards the same destination. The evacuation time is $1 + 2\pi/3 + \sqrt{3} + (\alpha-1)(1 + 2\pi/3 + \sqrt{3} - w)$.
    \end{description}
\end{proof}
\begin{theorem}\label{thm:wgreaterthanlb}
	The lower bound for evacuation with crash fault for crash time $w \in [1, 1 + 2\pi/3]$ is
	$
	\max_{w-1\leq t \leq 2\pi - 2(w -1)}\big( 1+ t + 2(\alpha+1)\cos(t/4)\big)
	$
	where $1 + t$ is the time for finding the exit and $w -1$ is the part of the perimeter explored by the crashed robot.
\end{theorem}
\begin{proof}
The lower bounds we describe are irrespective of the points where the robots hit the perimeter starting from the center of the disk. 
There can be two cases depending on the relation between crash time $w$ and the time an exit is found $(1 + t)$ by one of the robots.
\textbf{Case 1:} If a robot crashes after the exit is found, i.e., $ w \geq 1 + t$, then the robot not near the exit position is already travelling towards the exit. So, the worst-case time of evacuation would happen if and only if the robot crashes at the moment the exit is found, i.e., $w = 1 + t$. 
The worst-case distance from the crash position to exit is the diameter. If the robot covers part of the perimeter around the antipodal position ($C'$) of the crashed robot position ($C$), then the maximum distance between the crashed position and exit is at most $\overline{CB} = \overline{AC} = 2\sin((2\pi - (w -1))/4)$ as shown in Fig.~\ref{fig:lb1plust}.
The worst-case evacuation time is $w + 2(\alpha + 1)\sin((2\pi - (w -1))/4)$. 
\vspace{1em}

\noindent\begin{minipage}{\linewidth}
	\begin{minipage}{0.5\linewidth}
	\textbf{Case 2:} Suppose the exit is found at time $1 + t > w$. By the time the faulty robot crashes at $w$, it has covered at most $w-1$ on the perimeter. Consider two points $A$ and $B$ as shown in Fig.~\ref{fig:lb1plust} which are at a distance $w -1 + \epsilon$ from the crashed position $C$ for some small value of $\epsilon > 0$. If the adversary places the exit in the $\widearc{BA}$, then it would take at least $2\pi - 2(w-1)$ to explore the arc. The distance from any point on the arc is greater than $2\sin((w-1)/2)$. The time required for evacuation is at least 
    $$
        1 + t + 2(\alpha+ 1)\sin((w-1)/2)
    $$
	
	\end{minipage}\hspace{0.05\linewidth}
	\begin{minipage}{0.45\linewidth}\vspace{-4em}
		\begin{figure}[H]
			\centering
			\includegraphics[height=0.7\linewidth]{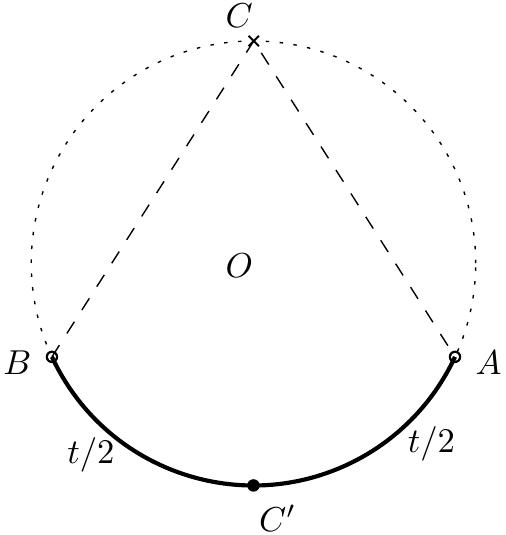}
			\caption{The exit lies in $\widearc{BA}$ with respect to crash position $C$}
			\label{fig:lb1plust}
		\end{figure}
	\end{minipage}
\end{minipage}\vspace{1em}
    For a value of $t < 2\pi - 2(w-1)$, this still holds. Then the time for evacuation is
    \begin{equation}\label{eq:lb1plust}
        1 + t  + 2(\alpha +1)\cos(t/4)
    \end{equation}
    The above expression holds as a lower bound since the robot which finds the exit at time $ 1+ t$ has to be at the exit and we show the linear distance between exit and crash position is at least $2\cos(t/4)$. As the robot has already crashed, the optimal path is to go to the crashed robot and pick up along the chord.
\end{proof}

\begin{remark}
    For $w -1 = 2\pi/3$, expression~\ref{eq:lb1plust} results in evacuation time $ 1 + 2\pi/3 + (\alpha+ 1)\sqrt{3}$ for $t = 2\pi/3$.
\end{remark}

\section{Upper Bound for Wireless Communication}\label{sec:upperbound}
The upper bound for evacuation in the wireless communication model is $1+ 2\pi/3+ \sqrt{3}$. This is the worst-case evacuation time of the algorithm proposed by Czyzowicz et al.~\cite{CzyzowiczGGKMP14}. In this paper, we also present evacuation strategies which determine the upper bound with chauffeuring. First, we present a simple strategy to put a ceiling on the upper bound. Next, we present two algorithms which provide us a tighter upper bound with respect to the crash time.

\subsection{Trivial Upper Bound (Algorithm $\mathcal{A}_0$(\texttt{MoveTogether}))}\label{sec:trivial}
Since at most one robot can be faulty, a trivial strategy is to just move both robots along the same path on the perimeter of the disk. Even if one of the robots becomes faulty, the other robot can chauffeur it and continue its search along the perimeter until it finds the exit. Then both the robots can evacuate via the exit as shown in Fig.~\ref{fig:MoveTogether}.
\begin{minipage}[h]{\linewidth}
	\begin{minipage}[h]{0.5\linewidth}
 The time required for this evacuation algorithm is at most $1 + 2\pi$ in the case where both robots are free of fault until the evacuation, where the time required to reach the perimeter is 1 and to search the perimeter is $2\pi$. If a robot becomes faulty after time $w$ from the activation, then the execution of the algorithm requires the following time
\begin{equation}\label{eq:trivialwc}
	z_0  = w + \alpha(1+2\pi -w)
\end{equation}
\end{minipage}\hspace{1em}
\begin{minipage}[h]{0.45\linewidth}
	\begin{figure}[H]
		\centering
		\includegraphics[width=0.7\linewidth]{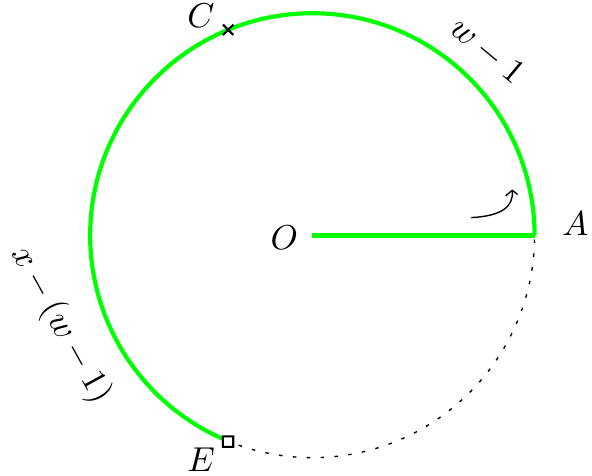}
		\caption{$R_1$ and $R_2$ start from $O$, hit the perimeter at $A$ and move together until they find exit at $E$ and evacuate.}
		\label{fig:MoveTogether}
		\vspace{1em}
	\end{figure}
\end{minipage}
\end{minipage}
where $z_0$ denotes the trivial upper bound with respect to the crash time $w$ and $\alpha \geq 1$.

\subsection{Evacuation Algorithm $\mathcal{A}_1$ (\texttt{MoveOpposite})}\label{sec:samepoint}
We start with a base algorithm which works with robots without faults. Both robots move together to an arbitrary point $A$ on the perimeter starting from the center $O$ and then move in opposite directions, i.e., clockwise and counter-clockwise as shown in Fig.~\ref{fig:carry}. Once a robot finds the exit, it sends a message to the other robot. On receiving the message, the other robot finds the position of the sender which is also the position of the exit by time of arrival of the message, knowledge of other robot's path and speed. Then it proceeds for evacuation along the straight line joining its current position and the exit.

As one of the robots can be faulty, it cannot communicate after it crashes.
To determine the position of the crashed robot, we assume that the robots communicate with each other constantly in very small intervals. If a robot crashes and then it fails to send a message, which determines the position of the crashed robot.
There can be two simple strategies in the aftermath of the crash. The non-faulty robot can carry the faulty robot and search together or the non-faulty robot searches for the exit without carrying the faulty robot. 
If both robots are travelling together, then they evacuate the moment the exit is found. Otherwise, the robot which finds the exit sends a message and both robots meet on the chord joining them. When the fault occurs, the non-faulty robot meets the faulty robot, then it chauffeurs the faulty robot to the exit. It increases the time required by a factor of $\alpha$. We describe the two strategies, where the exit is found after a robot has crashed in Section~\ref{sec:carry} and~\ref{sec:pickup}. An intermediate strategy, where the robot does not immediately pick up the robot after crash and searches for some distance on the perimeter, performs worse than the two aforementioned strategies. Please refer the appendix for more details.

\subsubsection{\texttt{SearchTogetherAfterCrash}}\label{sec:carry}
\begin{minipage}{\linewidth}	
\begin{minipage}[h]{0.6\linewidth}
	Two robots start together at the same time from the center $O$ as shown in Fig.~\ref{fig:carry}.
	Suppose the robot $R_1$ crashes at $C$ after a time $w$ at a distance $w-1$ along the arc from the point $A$, where the robots reach the perimeter of the disk, i.e., $\widearc{AC} = w-1$. $R_2$ is at $D$ when $R_1$ crashes at $C$. $R_2$ moves to $C$ along the chord $\overline{DC}$ and continues the search starting from $C$ on $\widearc{CD}$.
	Suppose the exit is located at $E$ and according to the convention, $A$ is the closest point in the clockwise direction from $E$ where a robot has reached the perimeter, so $\widearc{AE} = x$. By extension, $\widearc{CE} = x-w+1$.
	The time taken for the evacuation is $\overline{OA}+\widearc{DA}+\overline{DC}+\alpha\widearc{CE}$, i.e.,	
\end{minipage}\hspace{1em}
\begin{minipage}[h]{0.35\linewidth}\vspace{-3em}
	\begin{figure}[H]
	\centering
	\includegraphics[width= 0.9\linewidth]{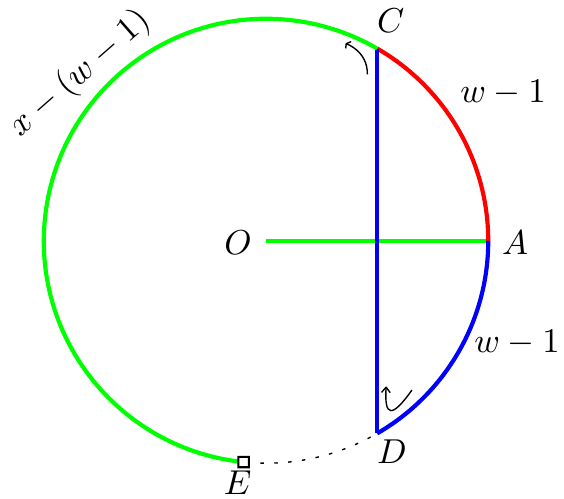}
	\caption{$R_1$ crashes at $C$ and then $R_2$ carries $R_1$ to continue searching for exit until $E$.}\label{fig:carry}\end{figure}
\end{minipage}
\end{minipage}
\begin{equation}
	w + 2\sin(w-1) + \alpha (x-w+1)
\end{equation}
The worst-case position of exit for which the evacuation time is maximum when $E$ is at an infinitesimally small distance from $D$ in the clockwise direction, i.e., $x = 2\pi - (w - 1)$. Then the worst-case evacuation time for this strategy is 
\begin{equation}
	z_{11} = w + 2\sin(w-1) + \alpha(2\pi-2(w-1))
\end{equation}
The worst-case evacuation time would be the maximum time over all possible values of $w$. For $w \in [1, 1+\pi]$, the critical point is obtained at $\partial z_{11}/\partial w = 0$, i.e.,
\begin{align}
	\partial z_{11} / \partial w = 1 - 2\alpha + 2\cos(w-1)  = 0
	\implies w = \arccos((2\alpha -1)/2) + 1
\end{align}
Since, $\partial^2 z_{11}/\partial w^2 = -2\sin(w-1) < 0 \,\forall w \in [1,1+\pi]$, the critical point is a local maximum. 
The worst-case evacuation time is obtained at $w -1 =\arccos(\alpha-1/2)$.
As the value of $\cos(w-1)$ ranges between $-1$ and $1$, the maximum value of $\alpha$ for which the critical point acts as the maximum is $(2\alpha -1)/2 = 1$, i.e., $\alpha = 1.5$.

\begin{remark}
	For $w\leq 1$, the \texttt{SearchTogetherAfterCrash} strategy has the same worst-case evacuation time as algorithm $\mathcal{A}_0$.
\end{remark}

\subsubsection{\texttt{SearchAloneAfterCrash}}\label{sec:pickup}

\begin{minipage}[h]{\linewidth}
	\begin{minipage}[h]{0.55\linewidth}
		Similarly, the robots $R_1$ and $R_2$ start at the center of the disk at the same time and move towards the perimeter. $R_1$ crashes at $C$. But $R_2$ continues to move along its own path until it finds the exit at $E$. Let $x$ be the distance along the arc to the exit from the point where the robots have hit the perimeter, i.e., $\widearc{AE}=x$. Then the distance from the crashed position of $R_1$ at $C$ to the exit at $E$ is $\widearc{CE} = x -(w -1)	$ as shown in Fig.~\ref{fig:pickup}. The time required for $R_2$ to reach the exit is $2\pi - x$. The time for evacuation would be,
	\end{minipage}\hspace{2ex}
	\begin{minipage}[h]{0.45\linewidth}\vspace{-4em}
		\begin{figure}[H]
			\centering
			\includegraphics[width=0.7\linewidth]{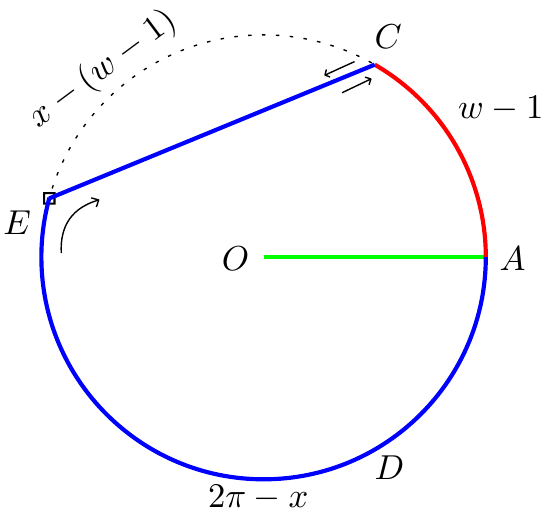}
			\caption{$R_1$ crashes at $C$, but $R_2$ continues searching for exit until $E$ and then picks up $R_1$ from $C$ and goes back to $E$.}\label{fig:pickup}\vspace{0.5em}
		\end{figure}
	\end{minipage}
\end{minipage}

\begin{equation}
	z_{12} = 1 + 2\pi - x +2(1+\alpha)\sin((x-(w -1))/2)
\end{equation}
Suppose the crash position of robot $R_1$ is fixed. Then the worst-case location of exit would be a critical point of $z_{12}$, i.e., $\partial z_{12}/\partial x = 0$.
\begin{align}
	(1+\alpha)\cos((x-(w-1))/2) = 1 \implies x = w -1 + 2\arccos(1/(1+\alpha))
\end{align}
Note that, $\partial z_{12} /\partial w < 0$, so $z_{12}$ is a monotonically decreasing function with respect to $w$.
Failure of a robot is not dependent on the position of the exit. 
So the variables $x$ and $w$ are independent. 
Hence, the worst-case evacuation time is obtained at $w = 1$. The worst-case evacuation time is $ 1 + 2\pi - 2\arccos(1/(1+\alpha)) + 2\sqrt{\alpha^2 + 2\alpha}$. Now let us determine the value of $\alpha$ for which the worst-case evacuation time of algorithm $\mathcal{A}_0$ exceeds the worst-case of \texttt{SearchAloneAfterCrash} strategy. 
The worst-case evacuation time of algorithm $\mathcal{A}_0$ from Equation~\ref{eq:trivialwc} is $ 1 + 2\pi\alpha$ for $w =1 $.
So,
\begin{align}
	&1 + 2\pi \alpha = 1 + 2\pi - 2\arccos(1/(1+\alpha)) + 2\sqrt{\alpha^2 + 2\alpha}\nonumber\\
	\implies&(1 + \alpha)\cos(\pi(1-\alpha)+\sqrt{\alpha^2 + 2\alpha}) -1 = 0 \label{eq:pickupalpha}
\end{align}
\sloppy The solution to Equation~\ref{eq:pickupalpha} is the value of $\alpha$ for which the worst-case of \texttt{SearchAloneAfterCrash} strategy coincides with algorithm $\mathcal{A}_0$. The corresponding $\alpha$ is 1.30346.
Hence, for values of $\alpha \leq 1.30346$, algorithm $\mathcal{A}_0$ has a better worst-case evacuation time compared to \texttt{SearchAloneAfterCrash} strategy. 

For the case where a robot becomes faulty before it reaches the perimeter, the non-faulty robot searches for the exit and then it picks up the faulty robot. Let $(w,0)$ be the position of faulty robot and $(\cos(x),\sin(x))$ be the position of the exit.
\begin{equation}
	z_{13} = 1 + x + (1+\alpha)\sqrt{(w-\cos(x))^2+\sin^2(x)}
\end{equation}	

\begin{remark}\label{rem:crashonline}
	If the exit is found before a robot crashes, then the worst-case would occur when it crashes on the perimeter.
\end{remark}
Since $R_1$ and $R_2$ start moving towards each other as soon as the exit is found and if $R_1$ finds the exit, then $R_2$ moves towards the exit moving along the line joining them. Hence, if $R_2$ crashes on the perimeter, that would result in the worst-case. Conversely, if $R_2$ finds the exit, it moves towards $R_1$ along the line joining them. $R_2$ would move at most to the midpoint of the line joining them. If $R_2$ becomes faulty at the midpoint of the line, then the evacuation time is less compared to when it failed at the perimeter.  

For $\alpha < 1.30346$, both strategies perform worse compared to the trivial algorithm at $w =0$. 
Both the strategies can be combined into one since the path of the robots remain the same until one of them crashes. 
Note that, the evacuation time has a local maximum at $ w =1 + \arccos((2\alpha -1)/2)$ for \sloppy\texttt{SearchTogetherAfterCrash} strategy for $\alpha < 1.5$, while it monotonically decreases for \texttt{SearchAloneAfterCrash} strategy. Since the evacuation time for $w=1$ is the same for both strategies at $\alpha = 1.30346$, the \texttt{SearchAloneAfterCrash} strategy performs better for all value of $w$.

Now, let us determine the value of $w$ for which both strategies perform the same, i.e., $z_{11} = z_{12}$. We have the following.
\begin{align}
	& w -1 + 2\sin(w-1) + \alpha (2\pi-2(w-1)) = 2\pi - 2\arccos(1/(1+\alpha)) - (w-1) + 2\sqrt{\alpha(\alpha+2)} \nonumber\\
	&\implies (w-1)(1-\alpha) + \sin(w-1) = \pi(1-\alpha) - \arccos(1/(1+\alpha)) +\sqrt{\alpha(\alpha+2)}\label{eq:combn}
\end{align}
For $\alpha < 1.30346$, we have the solution for $w$ from Equation~\ref{eq:combn}.
Let $\overline{w}$ be a solution to Equation~\ref{eq:combn}. If $w < \overline{w}$, then the other robot follows \texttt{SearchTogetherAfterCrash}, otherwise it follows \texttt{SearchAloneAfterCrash} strategy.
For $\alpha \geq 1.30346$, the \texttt{SearchAloneAfterCrash} strategy performs better than the  algorithm $\mathcal{A}_0$, where the worst-case of the \texttt{SearchTogetherAfterCrash} strategy is always greater than or equal to worst-case of algorithm $\mathcal{A}_0$. Hence, combination of the strategies will not yield a better result for $\alpha \geq 1.30346$.

\subsection{Evacuation Algorithm $\mathcal{A}_2$ (\texttt{MoveSameDirection})}\label{sec:zeta}

In this section, we describe an algorithm where the robots start from the center of the disk at the same time at an angle $\zeta$ with each other, where $0 \leq \zeta \leq 2\pi$. After reaching the perimeter, both robots start travelling in the counter-clockwise direction.
Without loss of generality, let us assume that $R_1$ crashes at time $w$.
We measure the angle $\zeta$ in the counter-clockwise direction from the faulty robot $R_1$.
\begin{figure}[h]
	\centering
	\includegraphics[width=0.7\linewidth]{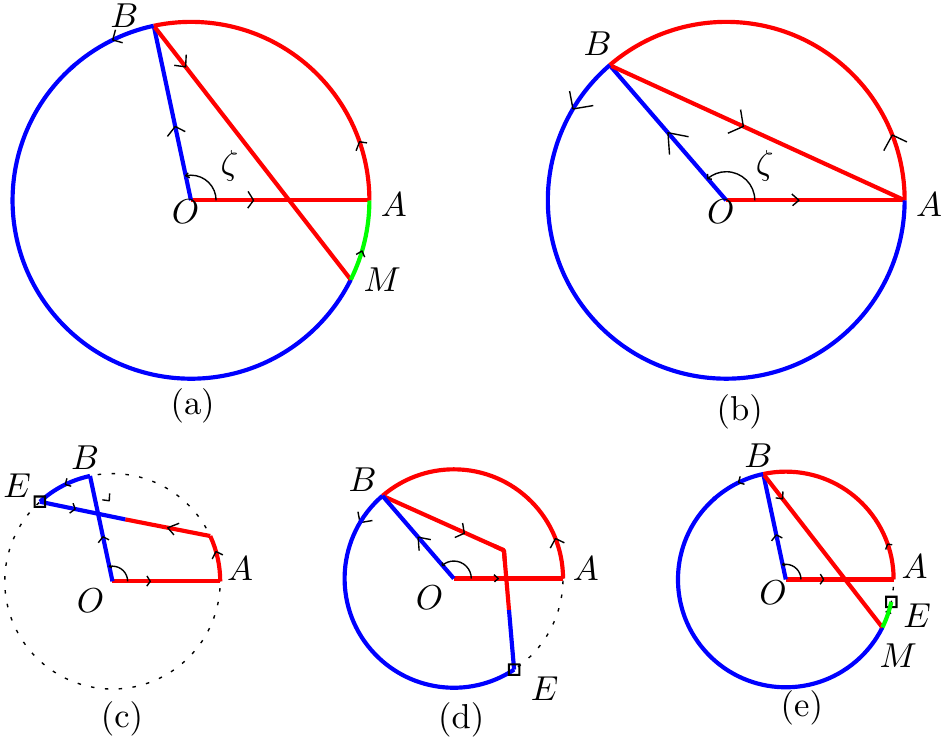}
	\caption{Path of the robots (red color for $R_1$ and blue for $R_2$) in Algorithm $\mathcal{A}_2$ without faults}
	\label{fig:zeta}
\end{figure}

The base strategy divides the perimeter into two arcs of length $\zeta$ and $2\pi - \zeta$. Each robot explores its own arc after reaching the perimeter from the center. For $\zeta < \pi$, $R_1$ would finish exploring its part of the arc. Then it will try to meet the other robot. Let us determine a meeting point $M$ (ref. Fig.~\ref{fig:zeta}(a)) such that $\widearc{AB} + \overline{BM} = \widearc{BM}$. The length, $m$, of $\overline{BM}$ can be determined from the following equation.
\begin{equation}
	\zeta + 2\sin((\zeta+m)/2) = \zeta + m
\end{equation} 
We can determine that the points $A$ and $M$ coincide for $\zeta = 2.24123$. So $R_1$ would move along $\overline{BM}$ for $\zeta < 2.24123$ (ref. Fig.~\ref{fig:zeta}(a)) and along $\overline{BA}$ for $\zeta \geq 2.24123$ (ref. Fig.~\ref{fig:zeta}(b)). For $\zeta > \pi$, a similar path is followed by $R_2$.

Now, we describe the action a robot takes if it finds an exit. On finding the exit, a robot messages the other robot. Then both robots start moving towards each other along the line joining them and meet at the midpoint (ref. Fig.~\ref{fig:zeta}(c) and (d)). Thenceforth both robots move towards the exit. If both robots are already travelling together, they evacuate via the exit (ref. Fig.~\ref{fig:zeta}(e)).

When $R_1$ crashes, it stops moving and stays there. Once $R_2$ has completely explored its own arc on the perimeter, then it moves to the position of $R_1$ and they explore the remaining part of the perimeter together if the exit is not found by that time.
Since we assume that $R_1$ is the robot which crashes, we claim that if the exit is found by $R_1$, then it would never result in the worst-case. The claim can be proved by the fact that if $R_2$ crashes, then the evacuation time would be higher. By symmetricity, there exists a situation where $R_1$ crashes corresponding to an exit found by $R_2$. So, in the following cases, we only explain the scenarios with $R_2$ finding the exit.

\begin{description}
	\item[Case 1:]	If $R_1$ crashes before the exit is found, then $R_2$ explores its own arc to search for the exit.
	\begin{itemize}
		\item If exit lies on $R_2$'s arc, then $R_2$ picks up $R_1$ from its crash position. The evacuation time is 
		\begin{equation}\label{eq:alphaplus1d}
			z_{21} = 1 + x - \zeta + (\alpha + 1)d
		\end{equation}
		where $d$ is the distance from the exit to the crash position and $x$ is the distance of the exit from the point where $R_1$ reaches the boundary in the counter-clockwise direction.
		The coordinates of exit position $E$ are $(\cos(x),\sin(x))$. The coordinates of the crash position $C$ are given by 
		\begin{align*}
			&(w, 0) &\textnormal{ for } w < 1 \\
			&(\cos(w -1), \sin(w - 1)) &\textnormal{ for } 1 \leq w \leq 1 + \zeta\\
			&(\lambda\cos(2\zeta+m)+(1- \lambda)\cos(\zeta), \lambda\sin(2\zeta+m)+(1 - \lambda)\sin(\zeta)) &\textnormal{ for } 1 + \zeta < w
		\end{align*}
		where $\lambda = \frac{w - \zeta -1}{m}$ if $\zeta < \pi$. Equation~\ref{eq:alphaplus1d} has two values of $x \in [\zeta, 2\pi]$, such that $\partial z_{21}/\partial x = 0$. The critical point may be a local maximum according to second derivative test. Hence the value of $x$ which provides the worst-case evacuation time for a given value of $w$ may be a critical point or a boundary point of domain of $x$ depending on the value of $\zeta$ and $\alpha$.
		\item If exit lies on $R_1$'s arc, $R_2$ picks up $R_1$ by travelling along a line joining the crash position and current position of $R_2$ and they explore the remaining part together. The evacuation time is 
		\begin{equation}
			z_{22} = w + d + \alpha(x - w - 1)
		\end{equation}
		where $d$ is the distance between crash position $(\cos(w-1), \sin(w-1))$ and position of $R_2$ at time $w$ determined similar to the previous case.
		In this case, the worst-case position of exit is just before the point where $R_2$ hits the perimeter, i.e., $x = \zeta - \epsilon$ for an $\epsilon \rightarrow 0$.
	\end{itemize}
	\item[Case 2:] If $R_1$ crashes after the exit is found, then $R_2$ is already moving towards $R_1$ along the line joining them. The evacuation time is 
	\begin{equation}
		z_{23} = w + (\alpha + 1)d
	\end{equation}
	where the distance $d$ between the crash position and position of $R_2$ at time $w$ determined similar to Case 1.
\end{description}
\section{Comparison between Algorithms}\label{sec:comparision}
We numerically evaluate the behavior of proposed algorithms. The worst-case evacuation time with respect to the crash time $w$ for specific values of $\alpha$ is evaluated.
$\{1, 1.30346,$ $ 1.5, 2\}$ are chosen as values of $\alpha$ since at these points behavior of the algorithms change according to the analysis in Section~\ref{sec:upperbound}.
The crash time $w$ is varied in $[0, 2\pi +1]$, because it takes $2\pi+1$ amount of time for a single non-faulty robot to evacuate from the disk. The evacuation times are evaluated at discrete values of $w$ with a gap of $\pi/120$. For each value of $w$, all possible worst-case positions of the exit are considered for Algorithms $\mathcal{A}_0$, $\mathcal{A}_1$ and $\mathcal{A}_2$. 
\noindent\begin{figure}[H]\centering
	\includegraphics[width=\linewidth]{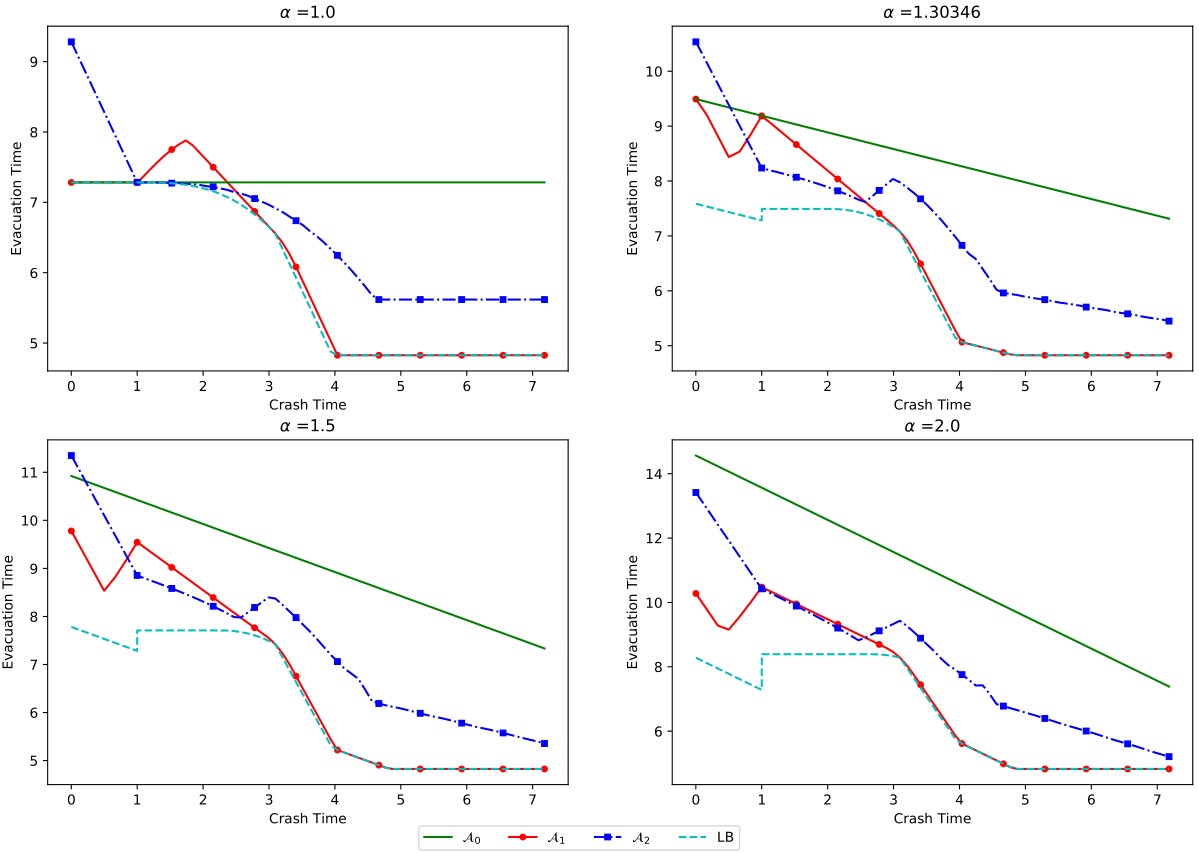}
	\caption{Comparison between algorithms and the lower bound for different values of $\alpha$}\label{fig:comparison}
\end{figure}\vspace{-2em}
As shown in Fig.~\ref{fig:comparison}, Algorithm~$\mathcal{A}_0$ performs better than $\mathcal{A}_1$ for $\alpha < 1.30346$.
For $w > 1 + 2\pi/2+\sqrt{3}/2$, Algorithm $\mathcal{A}_1$ is optimal. Observe that, $\mathcal{A}_2$ performs better compared to $\mathcal{A}_1$ for $w \in [1, 1+\pi/2]$.
The lower bound is dominated by $ t= 4\arcsin(2/(\alpha+1))$ for $\alpha = 2$ for crash time $ w < 1 + 2\pi/3$.
Also, for $\mathcal{A}_2$, the smallest worst-case evacuation time is found for discrete values of $\zeta \in [0,\pi]$ with a step size of $\pi/600$. We also evaluate the lower bound at the same values of $w$.
Fig.~\ref{fig:zetasim} shows the corresponding value of $\zeta$ for which the worst-case evacuation time is obtained for $\mathcal{A}_2$ and the four values of $\alpha$.
Observe that, for $\alpha> 1.30346$, $\zeta = \pi$ performs the best for algorithm $\mathcal{A}_2$ when crash time $w < 1 + 2\pi/3$. \vspace{-2em}
\begin{figure}[h]\centering
	\includegraphics[width=\linewidth]{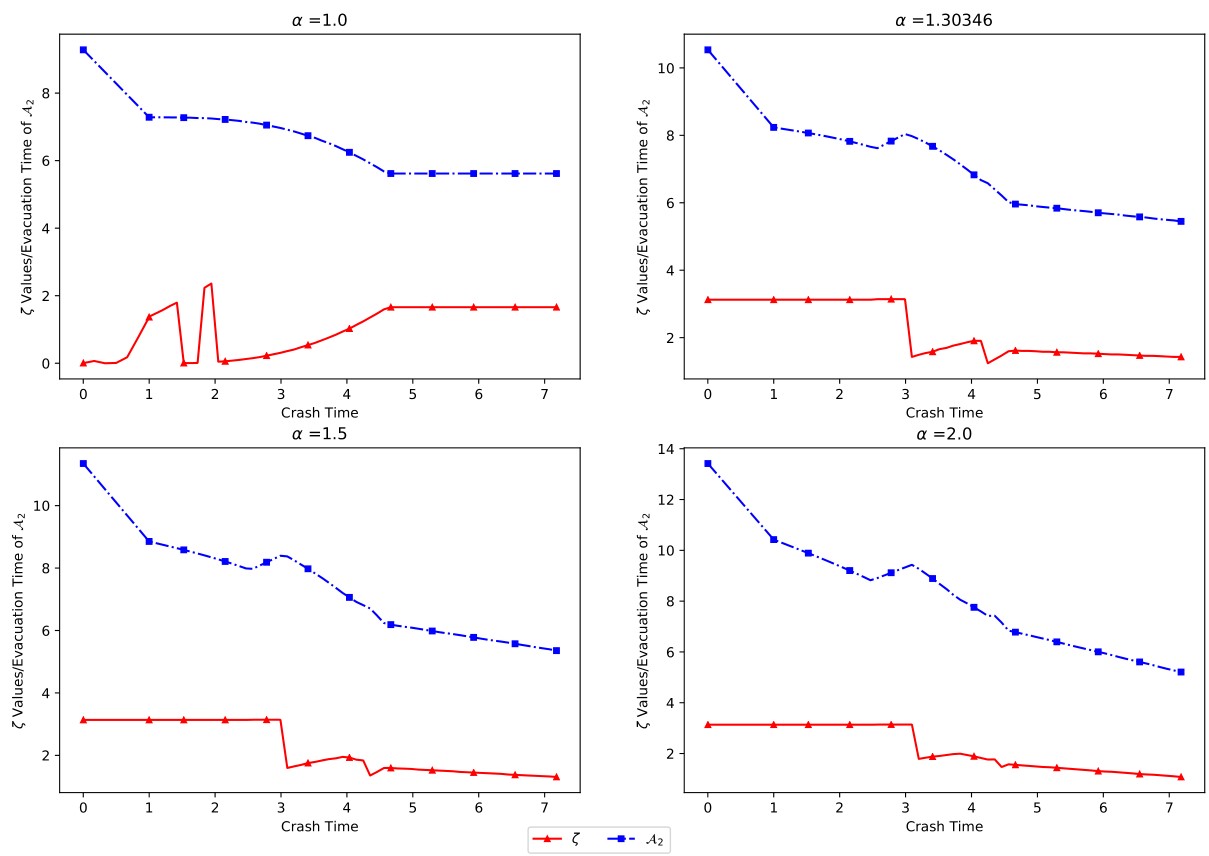}\vspace{-1em}
	\caption{Value of $\zeta$ with respect to the least worst-case evacuation time of $\mathcal{A}_2$}\label{fig:zetasim}
\end{figure}\vspace{-4em}
\section{Conclusion}\label{sec:conclusion}\vspace{-1em}
In this paper, we have introduced and analyzed evacuation algorithms for two robots out of which one can be faulty. 
Along with this, we also provide a lower bound for the evacuation time. For chauffeuring cost $\alpha =1$ the worst-case evacuation time is equal to the lower bound for crash time $w = 0$. This makes Algorithm $\mathcal{A}_0$ optimal. As the value of crash time $w$ increases beyond $1 + 2\pi/3$, the gap between the lower bound and evacuation time of  $\mathcal{A}_1$ is very small for all values of $\alpha$.
For $\alpha=1$, the lower bound is very close to the minimum evacuation time of $\mathcal{A}_1$ and $\mathcal{A}_2$, which says that the lower bound is tight. But the same does not happen for larger values of $\alpha$, where the lower bound is not very close.

This paper revisits the evacuation problem from a fault-tolerance aspect with one of the classical crash fault scenarios. The model can be further extended for generalized crash fault model with $k$ robots out of which $f$ are faulty. Further, it is interesting to design algorithms which can tighten the bounds presented in this paper.

\bibliographystyle{abbrv}
\bibliography{bib}

\begin{thebibliography}{10}

\bibitem{BrandtFRW17}
S.~Brandt, K.~Foerster, B.~Richner, and R.~Wattenhofer.
\newblock Wireless evacuation on m rays with k searchers.
\newblock In {\em {SIROCCO}, Porquerolles, France}, pages 140--157, 2017.

\bibitem{BrandtLLSW17}
S.~Brandt, F.~Laufenberg, Y.~Lv, D.~Stolz, and R.~Wattenhofer.
\newblock Collaboration without communication: Evacuating two robots from a
  disk.
\newblock In {\em {CIAC}, Athens}, pages 104--115, 2017.

\bibitem{ChuangpishitGS18}
H.~Chuangpishit, K.~Georgiou, and P.~Sharma.
\newblock Average case - worst case tradeoffs for evacuating 2 robots from the
  disk in the face-to-face model.
\newblock In {\em Algorithms for Sensor Systems - 14th International Symposium
  on Algorithms and Experiments for Wireless Sensor Networks, {ALGOSENSORS}
  2018, Helsinki, Finland, August 23-24, 2018, Revised Selected Papers}, pages
  62--82, 2018.

\bibitem{ChuangpishitMNO17}
H.~Chuangpishit, S.~Mehrabi, L.~Narayanan, and J.~Opatrny.
\newblock Evacuating an equilateral triangle in the face-to-face model.
\newblock In {\em {OPODIS}, Lisbon, Portugal}, pages 11:1--11:16, 2017.

\bibitem{CzyzowiczGGKMP14}
J.~Czyzowicz, L.~Gasieniec, T.~Gorry, E.~Kranakis, R.~Martin, and D.~Pajak.
\newblock Evacuating robots via unknown exit in a disk.
\newblock In {\em {DISC}, USA, October 12-15, 2014.}, pages 122--136, 2014.

\bibitem{CzyzowiczGGKKRW17}
J.~Czyzowicz, K.~Georgiou, M.~Godon, E.~Kranakis, D.~Krizanc, W.~Rytter, and
  M.~Wlodarczyk.
\newblock Evacuation from a disc in the presence of a faulty robot.
\newblock In {\em {SIROCCO}, Porquerolles, France}, pages 158--173, 2017.

\bibitem{fun/CzyzowiczGKKKNO18}
J.~Czyzowicz, K.~Georgiou, R.~Killick, E.~Kranakis, D.~Krizanc, L.~Narayanan,
  J.~Opatrny, and S.~M. Shende.
\newblock God save the queen.
\newblock In {\em 9th International Conference on Fun with Algorithms, {FUN}
  2018, June 13-15, 2018, La Maddalena, Italy}, pages 16:1--16:20, 2018.

\bibitem{sirocco/CzyzowiczGKKKNO18}
J.~Czyzowicz, K.~Georgiou, R.~Killick, E.~Kranakis, D.~Krizanc, L.~Narayanan,
  J.~Opatrny, and S.~M. Shende.
\newblock Priority evacuation from a disk using mobile robots - (extended
  abstract).
\newblock In {\em Structural Information and Communication Complexity - 25th
  International Colloquium, {SIROCCO} 2018, Ma'ale HaHamisha, Israel, June
  18-21, 2018, Revised Selected Papers}, pages 392--407, 2018.

\bibitem{CzyzowiczGKKNOS16}
J.~Czyzowicz, K.~Georgiou, E.~Kranakis, D.~Krizanc, L.~Narayanan, J.~Opatrny,
  and S.~M. Shende.
\newblock Search on a line by byzantine robots.
\newblock In {\em {ISAAC}, Sydney, Australia}, pages 27:1--27:12, 2016.

\bibitem{CzyzowiczGKNOV15}
J.~Czyzowicz, K.~Georgiou, E.~Kranakis, L.~Narayanan, J.~Opatrny, and
  B.~Vogtenhuber.
\newblock Evacuating robots from a disk using face-to-face communication
  (extended abstract).
\newblock In {\em {CIAC}, Paris, France}, pages 140--152, 2015.

\bibitem{CzyzowiczKKNO16}
J.~Czyzowicz, E.~Kranakis, D.~Krizanc, L.~Narayanan, and J.~Opatrny.
\newblock Search on a line with faulty robots.
\newblock In {\em {PODC}, Chicago, IL, USA}, pages 405--414, 2016.

\bibitem{KupavskiiW18}
A.~Kupavskii and E.~Welzl.
\newblock Lower bounds for searching robots, some faulty.
\newblock In {\em Proceedings of the 2018 {ACM} Symposium on Principles of
  Distributed Computing, {PODC} 2018, Egham, United Kingdom, July 23-27, 2018},
  pages 447--453, 2018.

\end{thebibliography}

\appendix
\newpage
\section{Comparison between strategies in Wireless Communication}
Two robots start together from the center $O$ at the same time as shown in Fig.~\ref{fig:carrypick}. Let $R_1$ be the robot that crashes at $C$. $R_2$ is at the point $D$ at the same time. $R_2$ travels up to the point $F$ at a distance $y$ along the arc before it decides to pick up $R_1$.

\noindent\begin{minipage}{\linewidth}
\begin{minipage}{0.45\linewidth}
Then the worst-case evacuation time appears if the exit is at a small distance $\epsilon (> 0)$ from $F$ in the unexplored part. The time for evacuation is $ f = w + y + 2\sin(w-1 + y/2) + \alpha(2\pi - 2(w-1) - y)$.
The function is $f$ is monotonically increasing with respect to $y$ up to $w + y/2 = \arccos(\alpha -1)$ and then monotonically decreases. So if $w > \arccos(\alpha-1)$, then it is better to maximize $y$. And if $w <= \arccos(\alpha-1)$, then it is better to have $y = 0$.
\end{minipage}\hspace{0.05\linewidth}
    \begin{minipage}{0.45\linewidth}
        \begin{figure}[H]
    \centering \vspace{-2em}
    \includegraphics[width=0.7\linewidth]{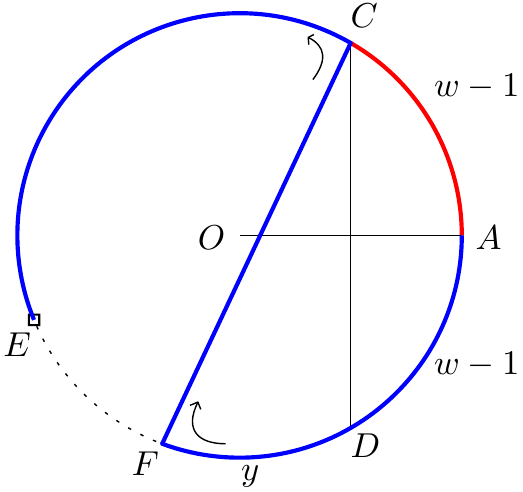}
    \caption{$R_2$ goes to pick $R_1$ at $F$ after travelling a distance $y$ and search along}
    \label{fig:carrypick}
\end{figure}
    \end{minipage}
\end{minipage}
\begin{remark}
    Going for meet-up immediately when the other robot fails is better compared to going for the crashed robot later if the exit is not found.
\end{remark}

\end{document}